\documentclass[11pt]{article}

\usepackage[margin=1.1in]{geometry}
\usepackage[T1]{fontenc}
\usepackage{lmodern}
\usepackage{microtype}
\usepackage{amsmath,amssymb,amsthm,mathtools,amsfonts,latexsym,amscd}
\usepackage{booktabs}
\usepackage{graphicx}
\usepackage{authblk}
\usepackage{float}
\usepackage{enumitem}
\usepackage[dvipsnames]{xcolor}
\usepackage[numbers,sort&compress]{natbib}
\usepackage{aliascnt}
\usepackage[colorlinks=true,linkcolor=blue!55!black,citecolor=green!40!black,urlcolor=blue!55!black]{hyperref}
\usepackage[nameinlink]{cleveref}
\usepackage{longtable}
\usepackage{placeins}
\newtheorem{thm}{Theorem}[section]

\newaliascnt{prp}{thm}
\newtheorem{prp}[prp]{Proposition}
\aliascntresetthe{prp}

\newaliascnt{lm}{thm}
\newtheorem{lm}[lm]{Lemma}
\aliascntresetthe{lm}

\newaliascnt{cor}{thm}
\newtheorem{cor}[cor]{Corollary}
\aliascntresetthe{cor}

\theoremstyle{definition}

\newaliascnt{df}{thm}
\newtheorem{df}[df]{Definition}
\aliascntresetthe{df}

\newaliascnt{remrk}{thm}
\newtheorem{remrk}[remrk]{Remark}
\aliascntresetthe{remrk}

\newaliascnt{exmp}{thm}

\aliascntresetthe{exmp}

\newaliascnt{conj}{thm}

\aliascntresetthe{conj}

\crefname{thm}{Theorem}{Theorems}
\crefname{prp}{Proposition}{Propositions}
\crefname{lm}{Lemma}{Lemmas}
\crefname{cor}{Corollary}{Corollaries}
\crefname{df}{Definition}{Definitions}
\crefname{remrk}{Remark}{Remarks}
\crefname{exmp}{Example}{Examples}
\crefname{conj}{Conjecture}{Conjectures}
\Crefname{thm}{Theorem}{Theorems}
\Crefname{prp}{Proposition}{Propositions}
\Crefname{lm}{Lemma}{Lemmas}
\Crefname{cor}{Corollary}{Corollaries}
\Crefname{df}{Definition}{Definitions}
\Crefname{remrk}{Remark}{Remarks}
\Crefname{exmp}{Example}{Examples}
\Crefname{conj}{Conjecture}{Conjectures}
\Crefname{enumi}{Theorem}{Theorems}
\crefname{enumi}{Theorem}{Theorems}
\numberwithin{equation}{section}

\newcommand{\Z}{\mathbb{Z}}
\newcommand{\Zk}{\mathbb{Z}_k}
\newcommand{\Evec}{\vec E}
\newcommand{\phiflow}{\varphi}

\title{A QUBO Formulation for Nowhere-Zero $k$-Flows}

\author[1,2,*]{Ali Lotfi \texttt{all054@usask.ca}}
\author[1,3,*]{Adam Carter \texttt{adam.carter@usask.ca}}
\author[4]{Mohammad Meysami \texttt{mohammad-meysami@utulsa.edu}}
\author[1,2]{Thuan Ha \texttt{thuan.ha@usask.ca}}
\author[1,2]{Kwabena Abrefa Nketia \texttt{kwabena.nketia@usask.ca}}
\author[1,2,*]{Steven J. Shirtliffe \texttt{steve.shirtliffe@usask.ca}}
\author[5,6,*]{Steven Rayan \texttt{rayan@math.usask.ca}}

\affil[1]{\normalsize Department of Plant Sciences, University of Saskatchewan, Saskatoon, SK, Canada}

\affil[2]{\normalsize Nutrien Centre for Digital and Sustainable Agriculture, University of Saskatchewan, Saskatoon, SK, Canada}

\affil[3]{\normalsize Crop Development Centre, University of Saskatchewan, Saskatoon, SK, Canada}

\affil[4]{\normalsize Department of Mathematics, The University of Tulsa, Tulsa, OK, USA}

\affil[5]{\normalsize Centre for Quantum Topology and Its Applications (quanTA), University of Saskatchewan, Saskatoon, SK, Canada}

\affil[6]{\normalsize Department of Mathematics and Statistics, University of Saskatchewan, Saskatoon, SK, Canada}

\affil[*]{\normalsize\textbf{Corresponding Authors}}

\begin{document}
\maketitle

\begin{abstract}

\color{black}
We consider the encoding of graph problems as Quadratic Unconstrained Binary Optimization (QUBO) problems, which are solvable by either quantum or classical annealers. Yet, the class of problems encoded as QUBO problems has not previously included nowhere-zero flows. Nowhere-zero flows are related to Tutte's $5$-flow conjecture and appear in many contexts in graph theory. We provide an encoding of nowhere-zero flows as a QUBO Hamiltonian and prove the correctness of the construction. Our construction yields a Hamiltonian $H_{\mathrm{mod},k}$ whose ground state has zero energy if and only if the graph $G$ has a nowhere-zero $\Zk$-flow. By Tutte's equivalence theorem, zero ground energy is equivalent to $\phiflow(G)\le k$, and the zero-energy degeneracy is given by the flow polynomial $F(G;k)$. In particular, when the ground-state energy is zero, this is also the ground-state degeneracy. The construction uses one-hot variables to represent the edge flow residues modulo $k$ and auxiliary variables to represent the per-vertex modular quotient. We prove that the correctness of the construction is independent of the choice of orientation, root vertex, and positive penalty weights. We verify the construction on $59$ examples of graphs and values of $k$ that include both yes-instances and no-instances. We exhaustively sweep orientations and root choices on selected robustness
instances and test a finite suite of positive penalty weights. The resulting Hamiltonian is implemented using the \texttt{dimod.BinaryQuadraticModel} class, which is compatible with the D-Wave Ocean SDK. Quantum-hardware runs and claims about potential speedup using these devices are left to follow-up work.

\end{abstract}

\pagebreak

\section{Introduction}
\label{sec:intro}

{\color{black}
Nowhere-zero flows are among the oldest objects in algebraic graph theory. \citet{Tutte1954} introduced nowhere-zero flows as a generalisation of map colourings and conjectured that every bridgeless multigraph admits a $5$-flow. Whilst this conjecture is still open after seventy years, the best general result is Seymour's $6$-flow theorem \cite{Seymour1981}, with a recent short proof by \citet{DeVosNurse2025}. Jaeger showed that whenever $G$ is $4$-edge-connected, it admits a $4$-flow \cite{Jaeger1979}. More recently, \citet{LovaszThomassenWuZhang2013} showed that every $6$-edge-connected graph admits a nowhere-zero $3$-flow. Integer flows and cycle covers were treated by \citet{Zhang1997}; we follow the notation of \citet{Diestel2025} throughout. Fix an orientation of the edges of $G$. A nowhere-zero $\Zk$-flow assigns each edge a nonzero residue in $\{1,\ldots,k-1\}$ such that the signed sum at every vertex is $0\pmod{k}$. By Tutte's equivalence theorem, this is equivalent to the integer $k$-flow defined in \cref{df:kflow}.

Despite its importance in graph theory, nowhere-zero flows has so far remained outside of the quantum-optimisation toolkit. Lucas et al. \cite{Lucas2014} provided QUBO encodings of a variety of NP-hard problems on graphs, such as the Hamiltonian cycle, graph colouring, vertex cover, and set packing problems. Nowhere-zero flows is one that is absent from this catalogue. To our knowledge, no QUBO encoding of the problem has been published.

\citet{ChandrasekaranLiuRavi2025} consider minimum-cost nowhere-zero flows and cut-balanced orientations, proving hardness and inapproximability results for this optimisation problem. The feasibility problem represents the decision problem at the heart of the proof. Thus, these results motivate the study of binary encodings of nowhere-zero flows, which can be used with QUBO-compatible optimisation methods. Such methods include quantum annealing as proposed by \citet{KadowakiNishimori1998}, adiabatic optimisation as proposed by \citet{FarhiGoldstoneGutmannSipser2000}, the Quantum Approximate Optimization Algorithm as proposed by \citet{FarhiGoldstoneGutmann2014}, and related digital or classical annealing methods.

\citet{EsperetLagoutteMarseloo2025} recently introduced the flow reconfiguration graph $\mathcal{F}(G,k)$ and asked when it is connected. Their questions depend computationally on the ability to enumerate or sample nowhere-zero $\Zk$-flows, for which a QUBO formulation provides a search space whose zero-energy states correspond exactly to the flows. The closest existing quantum-algorithmic result in the Tutte setting is the work of \citet{AharonovAradEbanLandau2007}, who give polynomial-time quantum algorithms for additive approximation of the Tutte polynomial of planar graphs. Their setting is approximation on planar graphs, whereas the present paper concerns exact feasibility on arbitrary nonempty loopless multigraphs.

The difficulty with modular conservation is that it is not itself a quadratic constraint: the congruence $R_v\equiv 0\pmod k$ does not have a native form in the space of QUBOs. Instead, at each non-root vertex, we introduce an auxiliary block of binary variables representing the signed quotient, which allows the conservation constraint to be encoded as an exact integer equality and the whole penalty to be quadratic. Theorem~\ref{thm:modular-qubo} establishes a bijection between nowhere-zero $\Zk$-flows on $G$ and the zero-energy states of $H_{\mathrm{mod},k}$, relates zero ground energy to the decision problem $\phiflow(G)\le k$, and proves that the emptiness and cardinality of the set of zero-energy states are invariant under the choice of orientation, root vertex, and positive penalty weights. Furthermore, Corollary~\ref{cor:flow-poly} shows that the size of the set of zero-energy states is equal to the flow polynomial $F(G;k)$, and Proposition~\ref{prp:gap} gives a lower bound $\min(A,B)$ on positive energies.

We verify the construction on a base test set of 59 (G,k) configurations, which include simple graphs, multigraphs with parallel edges, and both yes- and no-instances for the existence of a k-flow, ranging from graphs as small as the complete graph $K_3$ up to the Petersen graph. Furthermore, we test robustness under the three encoding choices appearing in the theorem, including exhaustive tests of
all possible orientations of $K_4$, $K_{3,3}$, and $\Theta_3$, exhaustive tests of all possible roots of $K_{3,3}$ and
the Petersen graph, and a test of different penalty weights on $K_4$.

}

The claim of the paper is limited to the reduction itself. A zero-energy state that any solver (quantum or classical) returns is a witness to the fact that $\phiflow(G) \le k$. A solver that only returns positive energies is not a certificate of nonexistence of solutions. The Hamiltonian is not run on any quantum hardware in this work. No claims are made about the speedup enabled by using quantum hardware to evaluate the Hamiltonian. The experiments performed in this work are implementation checks for
correctness of the construction; the exact checks are supplemented by a
randomized non-flow test. The Tutte conjecture and the 4-flow conjecture of \citet{Tutte1966} remain open problems in graph theory, and are not addressed in this paper.

The construction extends the QUBO catalogue of \citet{Lucas2014} and \citet{GloverKochenbergerDu2019} by one entry. The one-hot scheme is shared with Lucas's QUBO for graph $k$-colouring. The novelty is the auxiliary signed-quotient block encoding modular conservation. The bridge obstruction (\Cref{prp:bridge}) is a classical result; we include a short proof for completeness because the same summation trick reappears in the QUBO correctness argument. In this paper, \Cref{sec:preliminaries} fixes notation, introduces orientations and signed incidence, proves the bridge obstruction, and reduces the integer formulation to its modular version via Tutte's theorem. \Cref{sec:qubo} constructs $H_{\mathrm{mod},k}$, proves the bijection between zero-energy states and nowhere-zero $\Zk$-flows, derives the flow-polynomial corollary, and gives a sharp lower bound on positive energies. \Cref{sec:numerical} reports the verification protocol, the correctness and robustness suite, and the empirical variable counts on representative graph families.

\section{Preliminaries}\label{sec:preliminaries}

\subsection{Graphs and oriented edges}

This section fixes notation and records the facts from flow theory used later. We follow Diestel~\cite{Diestel2025} throughout. The objects of 
study are multigraphs without loops.

\begin{df}[Loopless multigraph]
A loopless multigraph $G=(V,E)$ consists of a finite vertex set $V$ and 
a finite edge set $E$ together with an incidence map assigning to each 
$e \in E$ an unordered pair of distinct vertices, its endpoints. 
Parallel edges (distinct $e,e' \in E$ with the same endpoints) are 
permitted; loops (edges with equal endpoints) are not.
\end{df}

A flow assigns a signed value to each edge, and the sign depends on which 
end we measure from. To make this precise, we work with oriented edges. 
Each undirected edge has two oriented copies.

\begin{df}[Oriented edge]
If $e=xy$ is an undirected edge, its two oriented copies are $(e,x,y)$ 
and $(e,y,x)$. We write $\Evec$ for the set of all oriented edges of $G$.
\end{df}

We can now define the central object: a function on oriented edges that is
consistent with reversing direction and balances at every vertex.

\begin{df}[Circulation]
Let $H$ be an abelian group written additively. A function $f:\Evec\to H$
is a \emph{circulation} if it satisfies:
\begin{enumerate}[label=(F\arabic*),ref=F\arabic*]
  \item\label{cond:F1} $f(e,x,y)=-f(e,y,x)$ for every oriented edge $(e,x,y)$;
  \item\label{cond:F2} for every $v\in V$, $\displaystyle\sum_{(e,v,w)\in\Evec} f(e,v,w)=0$.
\end{enumerate}
\end{df}

The first condition is \emph{antisymmetry}: reversing the orientation of
an edge negates its value. The second is \emph{Kirchhoff conservation}:
the signed sum of flow values at every vertex is zero. These are
conditions~\eqref{cond:F1} and~\eqref{cond:F2} of Diestel~\cite[Ch.~6]{Diestel2025}.

A circulation may vanish on some edges. The nowhere-zero flows studied here do not.

\begin{df}[$H$-flow]
An \emph{$H$-flow} is a circulation $f:\Evec\to H$ with $f(\vec e)\neq 0$
for every $\vec e\in\Evec$.
\end{df}

The group $H$ can be anything, but the classical case takes $H=\Z$ and
bounds the magnitude of $f$.

\begin{df}[$k$-flow]
\label{df:kflow}
For an integer $k>1$, a \emph{$k$-flow} on $G$ is a circulation
$f:\Evec\to\Z$ satisfying
\[
  0<|f(\vec e)|<k \qquad\text{for every } \vec e\in\Evec.
\]
\end{df}

The \emph{flow number} $\phiflow(G)$ is the least $k$ for which $G$ has a
$k$-flow, and $\phiflow(G)=\infty$ if no such $k$ exists.

\subsection{Orientations and signed incidence}
\label{sec:orient-incidence}

The circulation $f:\Evec\to H$ is defined on both oriented copies of every
edge, but by~\eqref{cond:F1} its values on the two copies determine each
other: knowing $f$ on one copy determines it on the other. It is therefore
wasteful to store both, and for the QUBO constructions in later sections
we want one value per edge. We now make this reduction precise.

First, we need a name for a choice of one copy per edge.

\begin{df}[Orientation]
An \emph{orientation} of $G$ is a function $D:E\to\Evec$ with
$D(e)\in\{(e,x,y),(e,y,x)\}$ for every edge $e=xy$.
\end{df}

Equivalently, an orientation is a subset of $\Evec$ containing exactly one of the two oriented copies of each edge.

From now on, fix an orientation $D$ of $G$. For each edge $e\in E$, write $t(e)$ and $h(e)$ for the tail and head of $D(e)$; that is, $D(e)=(e,t(e),h(e))$. Storing one value per edge means working with the single triple $D(e)$ rather than both oriented copies, and we need to rewrite Kirchhoff conservation in those terms. The bookkeeping is done by the sign of how each edge meets a vertex.

\begin{df}[Signed incidence]
\label{df:signed-incidence}
The \emph{signed incidence} of edge $e\in E$ at vertex $v\in V$ is
\[
\sigma_{v,e}=
\begin{cases}
+1, & v=t(e),\\
-1, & v=h(e),\\
0, & \text{otherwise.}
\end{cases}
\]
\end{df}

We say $e$ is \emph{incident to} $v$ if $v\in\{t(e),h(e)\}$, equivalently
if $\sigma_{v,e}\ne 0$. So $\sigma_{v,e}$ is $+1$ if $e$ leaves $v$, $-1$
if $e$ enters $v$, and $0$ if $e$ is not incident to $v$. We can now
rewrite Kirchhoff.

\begin{prp}[Fixed-orientation form of Kirchhoff]
\label{prp:fixed-orient}
Let $f:\Evec\to H$ be a circulation and $D$ an orientation of $G$. Then
\[
  \sum_{e\in E}\sigma_{v,e}\, f(D(e)) \;=\; 0
  \qquad\text{for every } v\in V.
\]
\end{prp}

\begin{proof}
Fix $v \in V$. Partition
\[
  E = T_v \sqcup H_v \sqcup O_v,
  \qquad
  T_v := \{e : t(e) = v\}, \quad
  H_v := \{e : h(e) = v\}, \quad
  O_v := \{e : v \notin \{t(e), h(e)\}\},
\]
so that
\[
  \sigma_{v,e} =
  \begin{cases} +1, & e \in T_v, \\ -1, & e \in H_v, \\ 0, & e \in O_v. \end{cases}
\]
For each edge $e$ incident to $v$, exactly one triple in $\Evec$ has $v$ in the tail slot:
\[
  e \in T_v \;\Longrightarrow\; (e, v, h(e)) = D(e),
  \qquad
  e \in H_v \;\Longrightarrow\; (e, v, t(e)) = \text{reverse of } D(e).
\]
Therefore
\begin{align*}
0
\;&\stackrel{\eqref{cond:F2}}{=}\;
  \sum_{(e,v,w) \in \Evec} f(e, v, w) \\
\;&=\;
  \sum_{e \in T_v} f(e, v, h(e))
  \;+\;
  \sum_{e \in H_v} f(e, v, t(e)) \\
\;&\stackrel{\eqref{cond:F1}}{=}\;
  \sum_{e \in T_v} f(e, v, h(e))
  \;-\;
  \sum_{e \in H_v} f(e, t(e), v) \\
\;&=\;
  \sum_{e \in T_v} f(D(e))
  \;-\;
  \sum_{e \in H_v} f(D(e)) \\
\;&=\;
  \sum_{e \in T_v} \sigma_{v,e}\, f(D(e))
  \;+\;
  \sum_{e \in H_v} \sigma_{v,e}\, f(D(e))
  \;+\;
  \sum_{e \in O_v} \sigma_{v,e}\, f(D(e)) \\
\;&=\;
  \sum_{e \in E} \sigma_{v,e}\, f(D(e)). \qedhere
\end{align*}
\end{proof}

\Cref{prp:fixed-orient} is the form of Kirchhoff we will use throughout
the rest of the paper. To keep the notation light, we write $f(e)$ for
$f(D(e))$ from now on; the two-orientation form $f(e,x,y)$ will not
appear again. For the fixed orientation $D$, we call the function $e\mapsto f(D(e))$
the \emph{circulation $f$ relative to $D$}. When this function is
nowhere zero on $E$, we call it the \emph{$H$-flow $f$ relative to $D$}.
Likewise, a nowhere-zero $\Zk$-flow or a $k$-flow relative to $D$ is a
function on $E$ obtained from a circulation on $\Evec$ by restricting to
the chosen orientation. For nowhere-zero $\Zk$-flows we always represent the nonzero residues by the set $\{1,\ldots,k-1\}\subset\Zk$. Equivalently, by \cref{prp:fixed-orient}, a
circulation relative to $D$ is a function $f:E\to H$ satisfying
\[
  \sum_{e\in E}\sigma_{v,e}\,f(e)=0
  \qquad\text{for every } v\in V.
\]

Having fixed $D$, one should ask whether the choice matters. Could
reversing some edges change whether $G$ admits a nowhere-zero flow? The
next lemma shows it cannot: flow existence is invariant under any
subset of edge reversals. This is a standard fact, but we give the
proof in detail because the bookkeeping it establishes will reappear in
later arguments.

\begin{lm}[Orientation independence]
\label{lm:orient}
Let $G$ be a loopless multigraph and let $D,D'$ be two orientations of
$G$. Let
\[
S:=\{e\in E : D'(e)\text{ is the reverse of }D(e)\}.
\]
For any abelian group $H$ and any circulation $f$ on $G$ relative to $D$,
define
\[
(\Phi_{D\to D'}f)(e)=
\begin{cases}
-f(e), & e\in S,\\
f(e),  & e\notin S.
\end{cases}
\]
Then $\Phi_{D\to D'}$ is a bijection from the circulations of $G$
relative to $D$ onto the circulations of $G$ relative to $D'$. It
restricts to bijections on nowhere-zero $H$-flows and on nowhere-zero
$k$-flows. In particular, the existence of either kind of flow is
independent of the chosen orientation.
\end{lm}

\begin{proof}
Write $\sigma$ and $\sigma'$ for the signed-incidence functions of
$D$ and $D'$, respectively. If $e\notin S$, then
\[
\sigma'_{u,e}=\sigma_{u,e}
\qquad\text{for all }u\in V.
\]
If $e\in S$, then
\[
\sigma'_{u,e}=-\sigma_{u,e}
\qquad\text{for all }u\in V.
\]
Let $f$ be a circulation relative to $D$ and set $f':=\Phi_{D\to D'}f$.
For every vertex $u\in V$,
\begin{align*}
\sum_{e\in E}\sigma'_{u,e}f'(e)
&=
\sum_{e\notin S}\sigma'_{u,e}f'(e)
+
\sum_{e\in S}\sigma'_{u,e}f'(e)\\
&=
\sum_{e\notin S}\sigma_{u,e}f(e)
+
\sum_{e\in S}(-\sigma_{u,e})(-f(e))\\
&=
\sum_{e\in E}\sigma_{u,e}f(e)\\
&=
0,
\end{align*}
so $f'$ is a circulation relative to $D'$.

If $f$ is nowhere zero, then $f'(e)$ is again nowhere zero on every edge,
so $\Phi_{D\to D'}$ restricts to a map on nowhere-zero $H$-flows. In the case $H=\Zk$, used throughout the rest of the paper, a reversed value $a$ in the representative set $\{1,\ldots,k-1\}$ is written as $k-a$.
For integer $k$-flows, the same edgewise sign change preserves the bound
\[
0<|f'(e)|=|f(e)|<k
\qquad\text{for all }e\in E,
\]
so $\Phi_{D\to D'}$ also maps nowhere-zero $k$-flows to nowhere-zero
$k$-flows.

Finally, the set $S$ is the same whether we pass from $D$ to $D'$ or
from $D'$ to $D$, so applying the construction twice restores the
original values. Thus
\[
\Phi_{D'\to D}\circ \Phi_{D\to D'}=\mathrm{id}
\qquad\text{and}\qquad
\Phi_{D\to D'}\circ \Phi_{D'\to D}=\mathrm{id}.
\]
Hence $\Phi_{D\to D'}$ and $\Phi_{D'\to D}$ are inverse bijections.
\end{proof}

\subsection{The bridge obstruction}

Before building any algorithm to find nowhere-zero flows, we should ask
when they exist at all. There is one local obstruction so simple it
rules the problem out instantly: a \emph{bridge}.

\begin{df}[Bridge]
A \emph{component} of $G$ is a maximal connected subgraph. An edge
$e\in E$ is a \emph{bridge} if deleting $e$ increases the number of
components of $G$.
\end{df}

Equivalently, a bridge is an edge whose removal splits $G$ into two nonempty vertex sets with no other edges between them. The following proposition shows that such an edge forbids every nowhere-zero flow:
no matter which group we pick, a bridge must carry value zero, which is exactly what nowhere-zero flows disallow.

\begin{prp}[Bridge obstruction]
\label{prp:bridge}
If $G$ has a bridge, then $G$ admits no nowhere-zero $H$-flow for any abelian group $H$. In particular $\phiflow(G)=\infty$.
\end{prp}

\begin{proof}
The strategy is to sum Kirchhoff's equation over one side of the cut
induced by the bridge, then show that every term cancels except the
one coming from the bridge itself.

Let $e_0=xy$ be a bridge, and write $X$ for the vertex set of the
component of $G-e_0$ containing $x$ and $\bar X=V\setminus X$ for its
complement. Since $e_0$ is a bridge, it is the only edge of $G$ with
one endpoint in $X$ and the other in $\bar X$.

Let $f$ be any circulation on $G$. We sum Kirchhoff's equation
(\cref{prp:fixed-orient}) over all vertices in $X$. Since each individual Kirchhoff sum is zero, so is the total:
\[
  S := \sum_{v\in X}\,\sum_{e\in E}\sigma_{v,e}\, f(e) \;=\; 0.
\]
We evaluate $S$ by swapping the order of summation and grouping edges by how many endpoints they have in $X$:
\[
  S \;=\; \sum_{e\in E}\,f(e)\,\underbrace{\sum_{v\in X}\sigma_{v,e}}_{=: \,c_e}.
\]
The coefficient $c_e$ depends only on the edge:
\[
  c_e \;=\;
  \begin{cases}
    (+1) + (-1) \;=\; 0, & \text{if } t(e),h(e)\in X,\\
    0, & \text{if } t(e),h(e)\in\bar X,\\
    +1, & \text{if } t(e)\in X,\;h(e)\in\bar X,\\
    -1, & \text{if } t(e)\in\bar X,\;h(e)\in X.
  \end{cases}
\]
The only edge with an endpoint in each of $X$ and $\bar X$ is the bridge $e_0$. Every other $c_e$ is zero, so the sum collapses to a single term:
\[
  S \;=\; c_{e_0}\, f(e_0) \;=\; \pm f(e_0).
\]
Combined with $S=0$, this forces $f(e_0)=0$. Thus every circulation on
$G$ vanishes on the bridge $e_0$, and no nowhere-zero $H$-flow can exist.

\end{proof}

{\color{black}
A nowhere-zero flow can exist only on a bridgeless graph, so
$\phiflow(G)=\infty$ whenever $G$ has a bridge. The QUBO construction
in \cref{sec:qubo}, however, is well-defined for every loopless
multigraph, including disconnected and bridged graphs. On a no-instance,
its ground-state energy is strictly positive; see \cref{rem:scope}.

}

\subsection{The modular version}

The definition of $k$-flow requires searching over signed integers $\{\pm 1, \ldots, \pm(k-1)\}$ with conservation in $\Z$. For a QUBO encoding it is more convenient to drop signs altogether and work modulo $k$: a smaller search space and a cleaner match for binary variables. The reduction to modular form is classical, due to Tutte.

Write $\Zk = \{0, 1, \ldots, k-1\}$ for the cyclic group of integers modulo $k$.

\begin{df}[Nowhere-zero $\Zk$-flow]
\label{df:zk-flow}
A \emph{nowhere-zero $\Zk$-flow} on $G$ is a function
$f:E\to\{1,\ldots,k-1\}\subset\Zk$ satisfying
\[
  \sum_{e \in E} \sigma_{v,e}\, f(e) \;\equiv\; 0 \pmod{k}
  \qquad\text{for every } v \in V.
\]
\end{df}

Comparing with \cref{df:kflow}: the edge values now live in 
$\{1, \ldots, k-1\}$ rather than $\{\pm 1, \ldots, \pm(k-1)\}$, and 
Kirchhoff conservation is required only modulo $k$. This reduction 
loses nothing:

\begin{thm}[Tutte equivalence]
\label{thm:tutte}
A multigraph $G$ admits a $k$-flow if and only if it admits a nowhere-zero
$\Zk$-flow.
\end{thm}

\begin{proof}
See \citet[Theorem~6.3.3]{Diestel2025}.
\end{proof}

\Cref{thm:tutte} is the bridge from graph theory to QUBO. The integer
formulation has signed variables and integer conservation; the modular
formulation has nonnegative residue variables and conservation modulo
$k$. Our QUBO constructions in \cref{sec:qubo} encode the modular form
directly.

\section{A modular QUBO for nowhere-zero \texorpdfstring{$k$}{k}-flows}
\label{sec:qubo}

{\color{black}Here, we will construct a Quadratic Unconstrained Binary Optimization (QUBO) Hamiltonian $H_{\mathrm{mod},k}$ whose ground-state energy is zero if and only if $G$ admits a nowhere-zero $\Zk$-flow (\cref{thm:modular-qubo}); \cref{thm:modular-qubo}(i) identifies the zero-energy states explicitly with the flows. By \cref{thm:tutte}, this is equivalent to $\phiflow(G)\le k$. In the rest of the paper we focus on numerical evaluation.}

\subsection{Encoding the two constraints}

A nowhere-zero $\Zk$-flow is specified by two requirements
(\cref{df:zk-flow}): each edge carries a residue in $\{1,\ldots,k-1\}$, and
the signed sum of those residues at every vertex is zero modulo~$k$. A QUBO
minimizes a quadratic polynomial in binary variables, so we need to express
both requirements as vanishing quadratic penalties.

The first requirement is a one-hot condition: for each edge $e$, exactly one
of $k-1$ indicator variables is on. This is standard in QUBO, and a squared
penalty of the form $\bigl(\sum_a x_{e,a}-1\bigr)^2$ vanishes precisely when
the sum is $1$.

The second requirement is modular conservation, which QUBO cannot express
directly: binary variables range over integers, not residues. The standard route is to replace the congruence $R_v \equiv 0 \pmod{k}$ with the integer equality $R_v = k q_v$ for an auxiliary integer quotient $q_v$, then encode $q_v$ in binary. A squared penalty $(R_v - k q_v)^2$ vanishes
exactly when the congruence holds, provided $q_v$ is allowed to take every
integer value in a range wide enough to contain the true quotient.

\subsection{Binary variables and decoded values}

The QUBO uses two blocks of binary variables: one-hot indicators for edge values and binary encodings of per-vertex quotients. We use the orientation $D$ of $G$ fixed
in \cref{sec:orient-incidence}, with signed incidences $\sigma_{v,e}$ as
in \cref{df:signed-incidence}. The Hamiltonian's polynomial form depends on $D$ through these signs; by
\cref{lm:orient} and the correctness theorem below, the existence of a
zero-energy state will not.

\begin{df}[One-hot edge variables]
\label{df:xea}
For each edge $e\in E$ and each residue $a\in\{1,\ldots,k-1\}$, let
\[
  x_{e,a}\in\{0,1\}
\]
be a binary variable. Write
\[
  x \;:=\; \{\, x_{e,a} \;:\; e \in E,\; 1 \le a \le k-1 \,\}
\]
for the collection of all such variables.
\end{df}

The variable $x_{e,a}$ records whether the flow on edge $e$ takes value $a$. The one-hot penalty below enforces that exactly one such variable is on for each edge. Once the one-hot penalty is satisfied, the decoded flow
value on edge $e$ is
\[
  F_e(x) \;=\; \sum_{a=1}^{k-1} a\, x_{e,a}.
\]
Summing $F_e$ with signed incidence gives the per-vertex residual.

\begin{df}[Vertex residual]
\label{df:Rv}
For a vertex $v\in V$, the \emph{signed residual} at $v$ is
\[
  R_v(x) \;=\; \sum_{e\in E} \sigma_{v,e}\, F_e(x).
\]
\end{df}

If the one-hot penalties are satisfied, $R_v(x)$ equals the Kirchhoff sum
of the decoded flow at $v$. Our goal is to require $R_v(x)\equiv 0\pmod k$
at every vertex.

We need a range for the auxiliary integer $q_v$. For any $x$ satisfying
the one-hot constraints, each decoded value $F_e(x)$ lies in
$\{1,\ldots,k-1\}$; writing $d(v)$, the \emph{degree} of $v$, for the
number of edges incident to $v$, we
then have
\[
  -(k-1)\, d(v) \;\le\; R_v(x) \;\le\; (k-1)\, d(v).
\]
Hence if $R_v(x)$ is divisible by $k$, the quotient $q_v$ lies in the
integer interval $[L_v, U_v]$, where
\[
  L_v := \left\lceil -\tfrac{(k-1)d(v)}{k}\right\rceil,
  \qquad
  U_v := \left\lfloor \tfrac{(k-1)d(v)}{k}\right\rfloor.
\]

One conservation constraint is redundant. Summing the signed residuals over all vertices gives
\[
  \sum_{v\in V} R_v(x)
  \;=\; \sum_{v\in V}\sum_{e\in E} \sigma_{v,e}\, F_e(x)
  \;=\; \sum_{e\in E} F_e(x) \sum_{v\in V}\sigma_{v,e}
  \;=\; 0,
\]
since each edge contributes $+1$ at its tail and $-1$ at its head, and
this identity holds for every assignment $x$. Assume $V\ne\varnothing$
for the construction, choose an arbitrary root $r\in V$, and impose
conservation only on
\[
  V^\star := V\setminus\{r\}.
\]
Then $R_v\equiv 0\pmod k$ for every $v\in V^\star$ implies
$R_r\equiv 0\pmod k$. If $G$ is disconnected, one could instead omit
one conservation equation in each connected component; we use the
single-root convention to keep the notation uniform.

We encode each non-root quotient in binary. For each $v\in V^\star$, let
\[
  B_v \;:=\; \lceil \log_2(U_v - L_v + 1)\rceil,
\]
so that $B_v$ is the smallest number of bits needed to represent every
integer in $[L_v, U_v]$.

\begin{df}[Quotient variables]
\label{df:Mv}
For each non-root vertex $v\in V^\star$, introduce binary variables
$p_{v,0},\ldots,p_{v,B_v-1}\in\{0,1\}$ and define
\[
  M_v(p) \;=\; L_v + \sum_{b=0}^{B_v-1} 2^b\, p_{v,b}.
\]
Write
\[
  p \;:=\; \{\, p_{v,b} \;:\; v \in V^\star,\; 0 \le b \le B_v - 1 \,\}
\]
for the collection of all such variables.
\end{df}

The map $(p_{v,0},\ldots,p_{v,B_v-1})\mapsto M_v(p)$ is a bijection onto
$[L_v, L_v + 2^{B_v} - 1]$, which contains $[L_v, U_v]$ but may extend
slightly beyond $U_v$ when $U_v - L_v + 1$ is not a power of two. 

This overshoot is harmless. If $R_v(x) = k\, M_v(p)$ for any integer $M_v$ in the representable range, then $R_v(x)\equiv 0\pmod k$.

{\color{black}
\begin{remrk}[Isolated vertices]
\label{rem:isolated}
If $d(v)=0$, then $L_v=U_v=0$, so
\[
  B_v=\left\lceil \log_2(U_v-L_v+1)\right\rceil
      =\left\lceil \log_2 1\right\rceil
      =0.
\]
Thus no quotient bits are introduced at $v$. In this case the quantity
$M_v(p)$ is interpreted as the constant $L_v=0$, and
\[
  R_v(x)=\sum_{e\in E}\sigma_{v,e}F_e(x)=0
\]
because no edge is incident to $v$. Hence the conservation term at $v$
vanishes identically:
\[
  (R_v(x)-kM_v(p))^2=(0-0)^2=0.
\]
The construction therefore extends without modification to loopless
multigraphs with isolated vertices.
\end{remrk}
}

\subsection{The Hamiltonian}

We can now define the Hamiltonian.

\begin{df}[Modular flow Hamiltonian]
\label{df:Hmodk}
Fix an orientation $D$, a root $r$, and positive constants $A,B$. Define
\begin{equation}
\label{eq:Hmodk}
H_{\mathrm{mod},k}^{D,r,A,B}(x,p)
  \;:=\;
  A\sum_{e\in E}\Bigl(\sum_{a=1}^{k-1} x_{e,a} - 1\Bigr)^{\!2}
  \;+\;
  B\sum_{v\in V^\star}\bigl(R_v(x) - k\, M_v(p)\bigr)^{2}.
\end{equation}
When $D$, $r$, $A$, and $B$ are fixed, we suppress the superscript and write
\[
H_{\mathrm{mod},k}(x,p).
\]
\end{df}

The first sum is the one-hot penalty, one term per edge. The second is the
modular conservation penalty, one term per non-root vertex. Every term is
a nonnegative square scaled by a positive constant, so
\[
H_{\mathrm{mod},k}^{D,r,A,B}(x,p)\ge 0,
\]
and zero energy requires every term to vanish simultaneously.

{\color{black}
We are now ready to state the main result of the paper, \Cref{thm:modular-qubo}. It describes the zero-energy
set of $H_{\mathrm{mod},k}$ completely: part~(i) exhibits an explicit
bijection with the nowhere-zero $\Zk$-flows of $G$, part~(ii) reads that
bijection as a decision procedure for $\phiflow(G)\le k$, and part~(iii)
records that the emptiness and cardinality of the zero-energy set depend
only on $G$ and $k$, not on the encoding choices $D$, $r$, $A$, and~$B$.
\begin{thm}[Structure of the modular QUBO]
\label{thm:modular-qubo}
Let $G$ be a nonempty loopless multigraph and let $k>1$. For the chosen
orientation $D$, root $r$, and positive penalty weights $A,B$, define
\[
  \mathcal{F}_k^{D}(G)
  :=
  \{f:E\to\{1,\ldots,k-1\}: f \text{ is a nowhere-zero }\Zk\text{-flow on }G \text{ relative to }D\},
\]
and
\[
  \mathcal{Z}_k^{D,r,A,B}(G)
  :=
  \{(x,p):H_{\mathrm{mod},k}^{D,r,A,B}(x,p)=0\}.
\]
Then:
\begin{enumerate}[label=\textup{(\roman*)},ref=\thethm(\roman*)]
\item \label{thm:mq-bij}\textbf{Bijection.}
The map
\[
    \Psi_{D,r,A,B}:\mathcal{F}_k^{D}(G)\to\mathcal{Z}_k^{D,r,A,B}(G),\qquad
  f\mapsto (x,p),
\]
defined by
\[
  x_{e,a}=\mathbf{1}[a=f(e)]
  \quad\text{for } e\in E,\ 1\le a\le k-1,
\]
and by
\[
  M_v(p)=\frac{R_v(x)}{k}
  \quad\text{for } v\in V^\star,
\]
is a bijection.

\item \label{thm:mq-dec}\textbf{Decision equivalence.}
The ground-state energy of $H_{\mathrm{mod},k}^{D,r,A,B}$ is zero if and only if
\[
  \mathcal{F}_k^{D}(G)\neq\varnothing,
\]
equivalently if and only if
\[
  \phiflow(G)\le k.
\]

\item \label{thm:mq-ind}\textbf{Parameter independence of emptiness and cardinality.}
For every admissible choice of $D$, $r$, $A$, and $B$, the corresponding
zero-energy set is in bijection with $\mathcal{F}_k^{D}(G)$. Consequently,
the emptiness and cardinality of
\[
  \mathcal{Z}_k^{D,r,A,B}(G)
\]
depend only on $G$ and $k$.
\end{enumerate}
\end{thm}

\begin{proof}

For part (i), let $f\in\mathcal{F}_k^{D}(G)$. Define the edge variables by
\[
  x_{e,a}:=\mathbf{1}[a=f(e)].
\]
Then every one-hot term vanishes. Because $f$ is a nowhere-zero $\Zk$-flow,
each non-root residual $R_v(x)$ is divisible by $k$. For each
$v\in V^\star$, set
\[
  q_v:=R_v(x)/k.
\]
By the residual bound we showed above, each $q_v$ lies in $[L_v,U_v]$. By \cref{df:Mv}, there is a unique
quotient-bit pattern $p_v$ with $M_v(p)=q_v$. Hence
\[
  (x,p)\in\mathcal{Z}_k^{D,r,A,B}(G),
\]
so $\Psi_{D,r,A,B}$ is well defined.

Conversely, let $(x,p)\in\mathcal{Z}_k^{D,r,A,B}(G)$. Since the one-hot
block vanishes, each edge $e$ selects a unique residue
$f(e)\in\{1,\ldots,k-1\}$. Thus
\[
  f:E\to\{1,\ldots,k-1\}
\]
is well defined and nowhere zero. Since every conservation term vanishes,
\[
  R_v(x)=k\,M_v(p)
  \qquad (v\in V^\star),
\]
so
\[
  R_v(x)\equiv 0 \pmod{k}
  \qquad (v\in V^\star).
\]
Since
\[
\sum_{v\in V}R_v(x)=0
\]
for every binary assignment $x$, and since
\[
R_v(x)\equiv 0 \pmod{k}
\qquad\text{for all }v\in V^\star,
\]
it follows that
\[
R_r(x)\equiv 0 \pmod{k}.
\]
Hence
$f\in\mathcal{F}_k^{D}(G)$.

Moreover, for each $v\in V^\star$, the integer $R_v(x)/k$ lies in
$[L_v,U_v]$, and the binary map of \cref{df:Mv} assigns a unique bit
pattern to that value. Therefore the quotient bits are uniquely
determined by $f$, so
\[
  (x,p)=\Psi_{D,r,A,B}(f).
\]
This proves part (i).

Part (ii) is immediate from part (i) together with \cref{thm:tutte}.

For part (iii), fix an orientation $D$. By part (i), every admissible choice of root $r$ and positive penalty weights $A,B$ yields a zero-energy set in bijection with the same set $\mathcal{F}_k^{D}(G)$. It therefore remains only to compare different orientations. By \cref{lm:orient}, for any other orientation $D'$, negating the value on each edge where $D$ and $D'$ disagree gives a bijection
\[
\mathcal{F}_k^{D}(G)\longrightarrow \mathcal{F}_k^{D'}(G).
\]
Hence the emptiness and cardinality of the flow set are independent of $D$, and so the emptiness and cardinality of the corresponding zero-energy sets depend only on $G$ and $k$.
\end{proof}
Part~(i) is a bijection, so it transfers any count of nowhere-zero
$\Zk$-flows directly to a count of zero-energy states. Applying the
classical flow-polynomial theorem on the flow side gives the following.
\begin{cor}[Flow polynomial]
\label{cor:flow-poly}
For every admissible choice of $D$, $r$, $A$, and $B$,
\[
  \left|\mathcal{Z}_k^{D,r,A,B}(G)\right|
  =
  F(G;k)
  =
  (-1)^{|E|-|V|+c(G)}\,T_G(0,1-k),
\]
where $c(G)$ is the number of connected components of $G$ and $T_G$ is the Tutte polynomial of $G$.
\end{cor}

\begin{proof}
By \cref{thm:modular-qubo}(i),
\[
  \left|\mathcal{Z}_k^{D,r,A,B}(G)\right|
  =
    \left|\mathcal{F}_k^{D}(G)\right|.
\]
The classical flow-polynomial theorem identifies
$\left|\mathcal{F}_k^{D}(G)\right|$ with $F(G;k)$; see \citet[Theorem~6.3.1 and Corollary~6.3.2]{Diestel2025}. The Tutte-polynomial specialization is standard.
\end{proof}

Consequently, any classical lower bound on the number of nowhere-zero
group flows immediately becomes a lower bound on the zero-energy
degeneracy of $H_{\mathrm{mod},k}$.
\Cref{thm:modular-qubo} and \Cref{cor:flow-poly} describe the zero-energy
set. We now turn to the energies above it: the next result shows that a
positive energy cannot be arbitrarily small, but is bounded below
by $\min(A,B)$.
\begin{prp}[Energy gap]
\label{prp:gap}
For any assignment $(x,p)$ with
\[
  H_{\mathrm{mod},k}^{D,r,A,B}(x,p)>0,
\]
one has
\[
  H_{\mathrm{mod},k}^{D,r,A,B}(x,p)\ge \min(A,B).
\]
The bound is best possible in general, but not attained on every
instance.
\end{prp}

\begin{proof}
In \cref{eq:Hmodk} each squared factor
\[
\Bigl(\sum_{a=1}^{k-1}x_{e,a}-1\Bigr)^2
\qquad\text{and}\qquad
\bigl(R_v(x)-kM_v(p)\bigr)^2
\]
is the square of an integer, hence a nonnegative integer, and therefore at least $1$ whenever it is nonzero. If $H_{\mathrm{mod},k}^{D,r,A,B}(x,p)>0$, at least one squared factor is nonzero. If it is a one-hot factor, the corresponding summand is at least $A$; if it is a conservation factor, the corresponding summand is at least $B$. All remaining summands are nonnegative, so in either case
\[
H_{\mathrm{mod},k}^{D,r,A,B}(x,p)\ge \min(A,B).
\]
The bound is attained, for example, on $\Theta_3$ at $k=2$: if $A\le B$, choose an assignment with exactly one violated one-hot factor and all conservation factors zero; if $B\le A$, choose an assignment with all one-hot factors zero and exactly one violated conservation factor. It is not attained on every instance; on $K_3$ at $k=2$ with $A=B=1$, exact enumeration gives smallest positive energy
\[
  2>1=\min(A,B).
\]
\end{proof}
}

\begin{remrk}[Scope of \cref{thm:modular-qubo}]
\label{rem:scope}
The hypotheses of \cref{thm:modular-qubo} require only that $G$ be a
nonempty loopless multigraph; connectedness and bridgelessness are not assumed.
The identity
\[
  \sum_{v\in V}R_v(x)=0
\]
holds for every orientation and every binary assignment $x$, because
each edge contributes $(+1)+(-1)=0$ to the total. Hence omitting one
root conservation equation is valid whether or not $G$ is connected.
A disconnected graph may still admit a nowhere-zero $\Zk$-flow if each component does. If $G$ has a bridge, or more generally if $G$ is a no-instance for the chosen value of $k$, the Hamiltonian remains well-defined and its ground-state energy is strictly positive.
\end{remrk}

\Cref{thm:modular-qubo} reduces the question $\phiflow(G)\le k$ to a
QUBO ground-state problem. The construction is uniform in $k$: each edge contributes $k-1$ one-hot variables, while each quotient block grows only logarithmically through $B_v$. In particular, doubling $k$ increases each $B_v$ by at most one bit. The next section validates \cref{thm:modular-qubo} by exact computation and quantifies the size of the resulting QUBO on representative
graph families.

\section{Numerical experiments}
\label{sec:numerical}

{\color{black}
We validate \cref{thm:modular-qubo} numerically and
characterize the Hamiltonian's size on benchmark graphs. \Cref{sec:protocol} defines the verification protocol: three implementation checks (C1)--(C3) covering forward encoding, sampled non-flow labellings, and direct enumeration within budget.
\Cref{sec:results} applies the protocol to a base test set of $59$
$(G,k)$ configurations across $16$ graphs, spanning simple graphs,
multigraphs with parallel edges, yes-instances, and no-instances.
\Cref{sec:results-robustness} then tests the parameter-independence of the emptiness and cardinality of the zero-energy set under the three theorem parameters $(A,B)$, $r$, and $D$ through exhaustive sweeps on $K_4$, $K_{3,3}$, $\Theta_3$, and the Petersen graph. 

\Cref{sec:results-vars} reports variable counts, coupler
counts, density, offset, and coefficient magnitudes on named
benchmarks, on $30$ random cubic and $33$ random $G(n,m)$ samples, and
on a dedicated snark family; see \cref{tab:snark-benchmarks}.
\Cref{sec:limitations} states the scope of these checks. Finally,
\cref{sec:results-heuristic} separates theorem-level invariance from
heuristic-solver behaviour, which can depend on $(A,B)$ and on the
instance. No quantum hardware was tested; a reference implementation
is available at \url{https://github.com/alilotfi90/nzflow-qubo}.

}

\subsection{Verification protocol}
\label{sec:protocol}

We implemented \cref{thm:modular-qubo} as a Python~3 package built
on the D-Wave Ocean SDK, with deterministic random seeds throughout.
The Hamiltonian is exposed as a \texttt{dimod.BinaryQuadraticModel},
which makes the construction directly compatible with D-Wave's
exact, simulated-annealing, and hardware backends. {\color{black}
We verified that the \texttt{neal} simulated-annealing sampler
recovers zero-energy flows on $K_4$ at $k=4$ and $K_{3,3}$ at $k=3$
with default hyperparameters. On the Petersen graph at $k=5$,
\texttt{neal} does not consistently locate zero-energy states at
default hyperparameters; sampler performance on that instance is
discussed separately in \cref{sec:results-heuristic}. 

The code builds $H_{\mathrm{mod},k}$ from the explicit formula in \cref{eq:Hmodk} and encodes a given $\Zk$-flow as the binary assignment $(x,p)$ predicted by the proof of \cref{thm:modular-qubo}. It also enumerates nowhere-zero $\Zk$-flows by two independent methods. The first is a spanning-tree extension following \citet[Ch.~6, Exer.~6]{Diestel2025}: iterate over the $k^{\beta}$ assignments of residues in $\Zk$ to the non-tree edges and propagate Kirchhoff conservation to the tree. Here $\beta:=|E|-|V|+c(G)$ is the cycle rank of $G$. The second is a direct enumerator that scans every labelling $f:E\to\{1,\ldots,k-1\}$ and tests modular conservation at every non-root vertex.}

For each $(G,k)$ configuration the script performs three checks: 
\textbf{(C1)} it encodes every nowhere-zero $\Zk$-flow on $G$ via 
the proof of \cref{thm:modular-qubo} and verifies 
$H_{\mathrm{mod},k}=0$; \textbf{(C2)} it samples random non-flow 
labellings and verifies the minimum energy over the quotient bits 
is positive; and (C3) when $(k-1)^{|E|}\le 2\times 10^7$, it directly enumerates all
edge labellings and verifies that the number of feasible labellings equals
the number of nowhere-zero $\Zk$-flows reported by the spanning-tree
verifier. By \cref{thm:modular-qubo}, these feasible labellings correspond
bijectively to zero-energy states. Any failure of these checks would indicate a mismatch between
the implementation and \cref{thm:modular-qubo} and would require
investigation.

\subsection{Theorem-correctness results}
\label{sec:results}

The principal test set comprises $59$ $(G,k)$ configurations
across $16$ graphs (complete and complete bipartite graphs, cycles, the cube $Q_3$, the Petersen graph, and several multigraph families with parallel edges), spanning $36$ simple-graph and $23$ multigraph configurations, and including both yes-instances and no-instances. For each graph we tested $k\in\{2,\ldots,k_{\max}\}$ where $k_{\max}\le 6$ is constrained
by the (C3) budget.

\textbf{Outcomes.} All $59$ pairs passed all applicable checks. C1 passed on the $49$ configurations with at least one enumerated flow and was vacuous on the $10$ no-flow configurations. C2 passed on all configurations with sampled non-flow labellings. Among the $58$ pairs for which C3 was within budget, the direct enumeration count matched the spanning-tree count in every case. The single C3 skip is the Petersen graph at $k=5$, where $4^{15}\approx 10^{9}$ exceeds the budget; all $240$ nowhere-zero $\Z_5$-flows still encoded to $H_{\mathrm{mod},k}=0$. The largest direct edge-labelling enumeration completed was $Q_3$ at $k=5$, examining $4^{12}\approx 1.7\times 10^7$ labellings and finding exactly $156$ feasible labellings, matching the verifier's $156$ nowhere-zero $\Z_5$-flows. By \Cref{thm:modular-qubo}(i), these correspond to $156$ zero-energy states. \Cref{tab:representative} reports representative cases.

\begin{table}[t]
\centering
\small
\caption{Representative exact flow counts and threshold instances. The flow counts are exact and use the cycle-space extension enumerator. Column C1 is marked for yes-instances where every exact flow encodes to zero energy under $H_{\mathrm{mod},k}$. Column \emph{C3 count} records the number of feasible edge labellings $f:E\to\{1,\ldots,k-1\}$ found by direct enumeration when within the budget $(k-1)^{|E|}\le 2\times 10^7$; by \Cref{thm:modular-qubo}, this equals the zero-energy-state count because the quotient bits are uniquely determined. A dash (--) in the C3 count column indicates that direct enumeration was
outside budget; a dash in C1 indicates that there were no enumerated
flows to check.}
\label{tab:representative}
\begin{tabular}{lrrrrrrr}
\toprule
Graph & $k$ & $|E|$ & $\beta$ & \#flows & C1 & vars & C3 count \\
\midrule
$K_3$ & 2 & 3 & 1 & 1 & yes & 7 & 1 \\
$K_4$ & 3 & 6 & 3 & 0 & -- & 21 & 0 \\
$K_4$ & 4 & 6 & 3 & 6 & yes & 27 & 6 \\
$K_{3,3}$ & 3 & 9 & 4 & 2 & yes & 33 & 2 \\
$\Theta_3$ & 3 & 3 & 2 & 2 & yes & 9 & 2 \\
Triangular prism & 3 & 9 & 4 & 0 & -- & 33 & 0 \\
Petersen & 3 & 15 & 6 & 0 & -- & 57 & 0 \\
Petersen & 4 & 15 & 6 & 0 & -- & 72 & 0 \\
Petersen & 5 & 15 & 6 & 240 & yes & 87 & -- \\
$K_4$ doubled & 3 & 12 & 9 & 176 & yes & 36 & 176 \\
$Q_3$ & 5 & 12 & 5 & 156 & yes & 69 & 156 \\
\bottomrule
\end{tabular}
\end{table}

\subsection{Robustness}
\label{sec:results-robustness}
\Cref{thm:modular-qubo}(i) gives a bijection between nowhere-zero $\Zk$-flows and zero-energy states, and \Cref{thm:modular-qubo}(iii) shows that the emptiness and cardinality of the zero-energy set are independent of the choice of orientation, root, and positive penalty weights. We test those three encoding choices here. On $K_4$ at $k=4$, six pairs of penalty weights
spanning two orders of magnitude in each direction
($(A,B)\in\{(1,1),(10,1),(1,10),(0.1,1),(1,0.1),(5,0.5)\}$) all
produced exactly $6$ zero-energy states matching the $6$ nowhere-zero
$\Z_4$-flows of $K_4$. On $K_{3,3}$ at $k=3$, the suite was rerun
with each of the $6$ vertices in turn as the chosen root, with
identical outcomes.

We exhaustively swept all $2^{6}=64$ orientations
of $K_4$ at $k=4$, all $2^{9}=512$ orientations of $K_{3,3}$ at
$k=3$, and all $2^{3}=8$ orientations of $\Theta_3$ at $k=3$; in
every one of these $584$ orientations the spanning-tree verifier
returned the predicted flow count exactly, confirming the
orientation-independence guaranteed by \cref{lm:orient}. All $10$ root choices on Petersen at $k=5$ produced $240$
nowhere-zero $\Z_5$-flows encoding to $H_{\mathrm{mod},k}=0$. Together these sweeps test the three parameter choices appearing in the theorem at the level reported in \cref{tab:robustness-summary}: weight and root sweeps compare zero-energy counts for the constructed BQM, while orientation sweeps compare the induced flow counts. No deviation was found.

\begin{table}[t]
\centering
\small
\caption{Exact robustness sweeps. Orientation sweeps preserve the exact nowhere-zero flow count. Root and weight sweeps preserve the exact zero-energy count of the constructed BQM, consistent with the theorem-level invariance of $H_{\mathrm{mod},k}$.}
\label{tab:robustness-summary}
\begin{tabular}{l l r r r}
\hline
Sweep & Graph & $k$ & cases & invariant count \\
\hline
orientations & $K_4$ & 4 & 64 & 6 (yes) \\
orientations & $K_{3,3}$ & 3 & 512 & 2 (yes) \\
orientations & $\Theta_3$ & 3 & 8 & 2 (yes) \\
roots & $K_{3,3}$ & 3 & 6 & 2 (yes) \\
roots & Petersen & 5 & 10 & 240 (yes) \\
weights & $K_4$ & 4 & 6 & 6 (yes) \\
\hline
\end{tabular}
\end{table}

\subsection{Empirical size and structure of the BQM}
\label{sec:results-vars}
{\color{black}This subsection records the size of $H_{\mathrm{mod},k}$ on the benchmark suite. We give the variable-count formula, expand the coupler bound, then read off measured values for named and random graph families.}
The total binary-variable count for $H_{\mathrm{mod},k}$ is
\begin{equation}
\label{eq:varcount}
N_k(G) \;=\; (k-1)|E| \;+\; \sum_{v\in V^\star} B_v,
\qquad B_v=\bigl\lceil \log_2(U_v-L_v+1)\bigr\rceil.
\end{equation}
The first term grows linearly in $|E|$ and proportionally to $k-1$.
The second term depends on the degree distribution of $G$ in a way
that is not captured by $|E|$ and $|V|$ alone. For irregular graphs
$N_k(G)$ depends on the choice of root $r$, since
$V^\star=V\setminus\{r\}$ excludes a different vertex for each
choice and $B_v$ depends on $d(v)$. Unless stated otherwise, all reported totals use $r$ as the highest-indexed vertex; on the irregular $G(n,m)$ samples below,
changing $r$ alters the total by at most two bits.
{\color{black}
\paragraph{Coupler counts, offset, and density.}
Variable count alone does not describe the size of a QUBO. Expanding
\cref{eq:Hmodk}, the one-hot block contributes at most
\[
  |E|\binom{k-1}{2}
\]
quadratic couplers, since each edge contributes a clique on its
$k-1$ residue variables. At a non-root vertex $v$, the conservation
block couples the $(k-1)d(v)$ incident edge-residue variables among
themselves, couples those edge-residue variables to the $B_v$ quotient
bits, and couples the quotient bits among themselves. Hence the number
of nonzero quadratic couplers is bounded by
\begin{equation}
\label{eq:couplerbound}
\#\mathrm{couplers}
\;\le\;
|E|\binom{k-1}{2}
+
\sum_{v\in V^\star}
\left[
\binom{(k-1)d(v)}{2}
+
(k-1)d(v)\,B_v
+
\binom{B_v}{2}
\right],
\end{equation}
with repeated couplers merging after expansion. The constant offset of
the binary quadratic model is
\begin{equation}
\label{eq:offset}
A|E| + Bk^2\sum_{v\in V^\star} L_v^2.
\end{equation}
Here the term $A|E|$ comes from the constant $+1$ in each
$\bigl(\sum_a x_{e,a}-1\bigr)^2$, and the term
$Bk^2\sum_{v\in V^\star}L_v^2$ comes from the constant part of
$\bigl(R_v(x)-kM_v(p)\bigr)^2$. We report the measured variable count,
quadratic-coupler count, density, and offset for representative
benchmarks in \cref{tab:qubo-structure} at $A=B=1$.

\begin{table}[t]
\centering
\small
\caption{Structural statistics of the binary quadratic models $H_{\mathrm{mod},k}$ at $A=B=1$. Density is the fraction of present quadratic couplers among all possible couplers on the logical graph. The column $\max|Q|$ reports the largest absolute final BQM coefficient after merging identical variable pairs across the one-hot and conservation blocks; the dynamic range divides this by the smallest nonzero quadratic coefficient.}
\label{tab:qubo-structure}
\begin{tabular}{lrrrrrrr}
\hline
Graph & $k$ & vars & couplers & density & offset & $\max |Q|$ & dyn. range \\
\hline
$\Theta_3$ & 3 & 9 & 36 & 1.000 & 39 & 144 & 72.0:1 \\
$K_4$ & 4 & 27 & 189 & 0.538 & 198 & 256 & 128.0:1 \\
$K_{3,3}$ & 3 & 33 & 174 & 0.330 & 189 & 144 & 72.0:1 \\
Triangular prism & 3 & 33 & 174 & 0.330 & 189 & 144 & 72.0:1 \\
$K_4$ doubled & 3 & 36 & 342 & 0.543 & 444 & 576 & 288.0:1 \\
Petersen & 4 & 72 & 558 & 0.218 & 591 & 256 & 128.0:1 \\
Petersen & 5 & 87 & 873 & 0.233 & 915 & 400 & 200.0:1 \\
$Q_3$ & 5 & 69 & 681 & 0.290 & 712 & 400 & 200.0:1 \\
\hline
\end{tabular}
\end{table}

\paragraph{Coefficient magnitudes.}
The one-hot block contributes quadratic coefficients of magnitude
\[
2A
\]
between distinct residue variables on the same edge. Expanding the
conservation block
\[
B\bigl(R_v(x)-kM_v(p)\bigr)^2
\]
produces three further families of quadratic coefficients. For distinct
one-hot variables at the same vertex, the coefficient magnitude is at
most
\[
  2B\,aa' \le 2B(k-1)^2.
\]
For a one-hot variable and a quotient bit, the magnitude is at most
\[
  2Bk\,a\,2^b \le 2Bk(k-1)2^{B_v-1}.
\]
For two distinct quotient bits, the magnitude is at most
\[
  2Bk^2\,2^{b+c} \le 2Bk^2\,2^{2(B_v-1)}.
\]
These three bounds are contributions from a single conservation block.
The final BQM coefficient on a pair is the sum of contributions from
every block in which the pair appears; in particular, a same-edge residue
pair $(x_{e,a}, x_{e,a'})$ receives the one-hot contribution $2A$ together
with both endpoints' conservation contributions when those endpoints are
non-root, giving a final coefficient up to $2A + 4B(k-1)^2$. In the
benchmark set the maximum $|Q|$ is dominated by the two-quotient-bit
contribution at a single vertex, which is unaffected by such merging.
Thus the conservation-block coefficient magnitudes grow quadratically in
$k$, exponentially in $B_v$, and linearly in $B$. The measured values
reported below are for $A=B=1$. In this benchmark set, the largest
measured coefficient occurs on the densest tested graph, $K_4$-doubled
at $k=3$, where $\max|Q|=576$ and the smallest nonzero quadratic
coefficient is $2$, giving a dynamic range of $288{:}1$. The Petersen
graph at $k=5$ is next, with $\max|Q|=400$ and a dynamic range
of $200{:}1$. For fixed-precision annealing hardware the BQM would have
to be normalized before solving, and alternative quotient encodings
remain a reasonable direction for future work.
}
\begin{table}[t]
\centering
\begin{tabular}{lrrrrrr}
\toprule
Graph & $n$ & $m$ & $\beta$ & $\phiflow$ & vars ($k$=4) & vars ($k$=5) \\
\midrule
$K_3$ & 3 & 3 & 1 & 2 & 13 & 16 \\
$K_4$ & 4 & 6 & 3 & 4 & 27 & 33 \\
$K_5$ & 5 & 10 & 6 & 2 & 42 & 52 \\
$K_{3,3}$ & 6 & 9 & 4 & 3 & 42 & 51 \\
$K_{3,4}$ & 7 & 12 & 6 & 3 & 54 & 66 \\
$C_5$ & 5 & 5 & 1 & 2 & 23 & 28 \\
$C_6$ & 6 & 6 & 1 & 2 & 28 & 34 \\
$Q_3$ & 8 & 12 & 5 & 3 & 57 & 69 \\
Petersen & 10 & 15 & 6 & 5 & 72 & 87 \\
$\mathrm{GP}(7,2)$ & 14 & 21 & 8 & 4 & 102 & 123 \\
$\mathrm{GP}(8,3)$ & 16 & 24 & 9 & 3 & 117 & 141 \\
\bottomrule
\end{tabular}
\caption{Exact flow numbers and variable counts for named benchmark graphs.}
\label{tab:benchmarks}
\end{table}

The named-benchmark family (\cref{tab:benchmarks}) ranges from
$K_3$ ($n=3,|E|=3$) to the generalised Petersen graph
$\mathrm{GP}(8,3)$ ($n=16,|E|=24$). Variable counts at $k=5$ range
from $16$ to $141$, with the Petersen graph itself contributing $87$.

On random connected bridgeless cubic graphs
(\cref{tab:random-cubic-summary}), $30$ samples spanning
$n\in\{6,\ldots,16\}$, every vertex has $d(v)=3$, so $L_v=-2$,
$U_v=2$, $B_v=3$ at both $k=4$ and $k=5$. The variable count
reduces to
\[
  N_4(G) = 3|E| + 3(|V|-1), \qquad N_5(G) = 4|E| + 3(|V|-1).
\]
Flow numbers in the sample fall in $\{3,4\}$. The threshold no-instances
in the present revision include Petersen at $k=3$ and $k=4$, together
with the triangular prism at $k=3$.

Snarks, the bridgeless cubic graphs that are not $3$-edge-colorable, are no-instances for $k\le 4$: for cubic graphs, $3$-edge-colorability is equivalent to admitting a nowhere-zero $4$-flow, and a nowhere-zero $3$-flow would induce a $4$-flow. They are exercised directly in \cref{tab:snark-benchmarks}. The suite comprises the Petersen graph and the Isaacs flower snarks $J_5$ ($20$ vertices) and $J_7$ ($28$ vertices), the latter members of an infinite parametric family; \texttt{flower\_snark(5)} is verified isomorphic to SageMath's canonical \texttt{graphs.FlowerSnark()}. These graphs are confirmed to be no-instances at $k=3$ and $k=4$ by a
combination of enumeration and structural reasoning: Petersen and $J_5$
are checked by exact enumeration, $J_7$ at $k=3$ is checked by exact
enumeration, and $J_7$ at $k=4$ follows from non-$3$-edge-colorability. The flower snark $J_5$ admits $16{,}200$ nowhere-zero $\Z_5$-flows, every one of which encodes to $H_{\mathrm{mod},k}=0$ (check~C1). A complete finite snark catalogue is not required for the correctness claim, and a systematic snark hardness survey remains future work.

\begin{table}[t]
\centering
\begin{tabular}{rrrrlrr}
\toprule
$n$ & $|E|$ & $\beta$ & samples & $\phiflow$ values & vars ($k=4$) & vars ($k=5$) \\
\midrule
6 & 9 & 4 & 5 & $\phiflow=3$ ($\times 3$), $\phiflow=4$ ($\times 2$) & 42 & 51 \\
8 & 12 & 5 & 5 & $\phiflow=3$ ($\times 1$), $\phiflow=4$ ($\times 4$) & 57 & 69 \\
10 & 15 & 6 & 5 & $\phiflow=3$ ($\times 2$), $\phiflow=4$ ($\times 3$) & 72 & 87 \\
12 & 18 & 7 & 5 & $\phiflow=4$ ($\times 5$) & 87 & 105 \\
14 & 21 & 8 & 5 & $\phiflow=4$ ($\times 5$) & 102 & 123 \\
16 & 24 & 9 & 5 & $\phiflow=4$ ($\times 5$) & 117 & 141 \\
\bottomrule
\end{tabular}
\caption{Random connected bridgeless cubic graphs, summarized by vertex count. Five samples were drawn at each $n$. Variable counts are constant within each row because every cubic vertex contributes the same number of auxiliary bits. No snarks occurred in this random cubic family; threshold no-instances such as Petersen at $k=3,4$ and the triangular prism at $k=3$ are listed separately in Table~\ref{tab:representative}.}
\label{tab:random-cubic-summary}
\end{table}

\begin{table}[t]
\centering
\begin{tabular}{rrrrlrr}
\toprule
$n$ & $|E|$ & $\beta$ & samples & $\phiflow$ values & vars ($k=4$) & vars ($k=5$) \\
\midrule
10 & 14 & 5 & 5 & $\phiflow=3$ ($\times 4$), $\phiflow=4$ ($\times 1$) & 64--66 & 79--80 \\
10 & 16 & 7 & 5 & $\phiflow=3$ ($\times 5$) & 71--73 & 88--90 \\
10 & 18 & 9 & 5 & $\phiflow=3$ ($\times 4$), $\phiflow=4$ ($\times 1$) & 79--81 & 99--100 \\
15 & 19 & 5 & 3 & $\phiflow=3$ ($\times 3$) & 89--91 & 109--110 \\
15 & 21 & 7 & 5 & $\phiflow=3$ ($\times 5$) & 98--99 & 119--121 \\
15 & 23 & 9 & 5 & $\phiflow=3$ ($\times 5$) & 104--106 & 129--130 \\
20 & 28 & 9 & 5 & $\phiflow=4$ ($\times 5$) & 133--134 & 161--163 \\
\bottomrule
\end{tabular}
\caption{Random $G(n,m)$ samples, summarized by $(n,m)$. Up to five samples were drawn at each pair and then filtered for connectivity and bridgelessness; some pairs retain fewer than five graphs after filtering. Variable-count ranges document that the auxiliary block depends on degree distribution, not just $(n,m)$.}
\label{tab:random-gnm-summary}
\end{table}

\begin{table}[H]
\centering
\small
\begin{tabular}{lrrrrrrrr}
\hline
Graph & $k$ & $n$ & $|E|$ & $\beta$ & \#flows & C1 & vars & couplers \\
\hline
Petersen & 3 & 10 & 15 & 6 & 0 & -- & 57 & 312 \\
Petersen & 4 & 10 & 15 & 6 & 0 & -- & 72 & 558 \\
Petersen & 5 & 10 & 15 & 6 & 240 & yes & 87 & 873 \\
Flower snark $J_5$ & 3 & 20 & 30 & 11 & 0 & -- & 117 & 657 \\
Flower snark $J_5$ & 4 & 20 & 30 & 11 & 0 & -- & 147 & 1173 \\
Flower snark $J_5$ & 5 & 20 & 30 & 11 & 16200 & yes & 177 & 1833 \\
Flower snark $J_7$ & 3 & 28 & 42 & 15 & 0 & -- & 165 & 933 \\
Flower snark $J_7$ & 4 & 28 & 42 & 15 & -- & -- & 207 & 1665 \\
Flower snark $J_7$ & 5 & 28 & 42 & 15 & -- & -- & 249 & 2601 \\
\hline
\end{tabular}
\caption{Snark benchmarks. The Petersen graph and the Isaacs flower
snarks $J_5$ (20 vertices) and $J_7$ (28 vertices) are connected,
bridgeless, cubic, girth at least $5$, and not $3$-edge-colorable;
\texttt{flower\_snark(5)} is verified isomorphic to SageMath's canonical
\texttt{graphs.FlowerSnark()}. Since for cubic graphs
$3$-edge-colorability is equivalent to admitting a nowhere-zero
$4$-flow, and a nowhere-zero $3$-flow would imply a nowhere-zero
$4$-flow, the $k\in\{3,4\}$ snark rows are no-instances. The $k=5$
rows probe the near-threshold case: Petersen and $J_5$ are confirmed
yes-instances by the enumerated counts in the table, while the $J_7,k=5$
enumeration is outside the stated budget. The column \#flows is the exact cycle-space enumerator count when the
enumeration is within budget; a dash (--) in the \#flows column marks
configurations whose enumeration $(k-1)^\beta$ exceeds the budget
$5\times 10^6$. The C1 column is marked yes when enumerated flows were
checked to encode to $H_{\mathrm{mod},k}=0$, and is marked -- when C1 is
vacuous or outside budget. The structural columns are exact for every
row. For $J_7,k=4$, the no-instance conclusion follows from
non-$3$-edge-colorability rather than from enumeration.}
\label{tab:snark-benchmarks}
\end{table}
\FloatBarrier

On $33$ random $G(n,m)$ samples at controlled cycle-rank~(\cref{tab:random-gnm-summary}), the variable count varies within
each $(n,m)$ pair. For instance, the $5$ samples at $(n,m)=(10,14)$
gave variable counts ranging from $64$ to $66$ at $k=4$. This
illustrates that the auxiliary block depends on the degree
distribution, not on $(n,m)$ alone. Across the $33$ samples, $26$
have $\phiflow=3$ and the remaining $7$ have $\phiflow=4$. Across both random families and the named benchmarks, the measured variable counts agree with \cref{eq:varcount}; variation within a fixed $(n,m)$ pair comes from the degree distribution through the values $B_v$.

\subsection{Scope and limitations}
\label{sec:limitations}

The theorem is proved analytically. The verification protocol provides exact implementation checks for (C1) and, within the enumeration budget, for (C3), while (C2) is a randomized non-flow check; several questions remain outside its scope. We mark them here. In the base 59-configuration theorem-correctness table, the C3 budget of $2\times 10^7$ edge labellings was exceeded only by Petersen at $k=5$, where C1 and C2 still passed. The snark table uses a separate budget of $5\times 10^6$ cycle-space assignments; the corresponding out-of-budget rows are marked by dashes in Table~7. No quantum hardware was tested: all energies reported are computed exactly from \cref{eq:Hmodk}. The behaviour of $H_{\mathrm{mod},k}$ under quantum or classical annealing, including embedding overhead on D-Wave Pegasus or Zephyr topologies and chain-strength selection, is left
for future work. The random-graph samples reach cycle-rank $\beta=9$
and vertex count $n=20$. These are exact correctness checks and
formulation-size measurements, not scaling claims. A reference
implementation is available at
\url{https://github.com/alilotfi90/nzflow-qubo}.

{\color{black}
\subsection{Theorem-level versus heuristic-level parameter sensitivity}
\label{sec:results-heuristic}

The parameter-independence in \Cref{thm:mq-ind} is a statement about the exact zero-energy set, not about the behaviour of a particular heuristic solver. Our exact penalty-weight sweep on $K_4$ at $k=4$
confirms the theorem-level statement: the six tested pairs
\[
(A,B)\in\{(1,1),(10,1),(1,10),(0.1,1),(1,0.1),(5,0.5)\}
\]
all produce the same ground-state manifold of six zero-energy states.

Heuristic samplers behave differently. Their success probability and
time-to-solution can depend strongly on the ratio $A/B$ and on the
instance itself. In our \texttt{neal} runs, $K_4$ at $k=4$ and
$K_{3,3}$ at $k=3$ reached zero energy readily across all tested
weight pairs. On the Petersen graph at $k=5$, by contrast,
\texttt{neal} stalled at strictly positive energies across the tested
grid
\[
  (A,B)\in\{1,3,5,10,50\}\times\{1,3,10,50\},
\]
with up to $5000$ reads and $50{,}000$ sweeps per run across all tested
random seeds. This behaviour does not contradict
\cref{thm:modular-qubo}; it reflects heuristic search difficulty on a
hard near-threshold instance. Determining the best-performing weight
ratios, annealing schedules, and hardware embeddings is a separate
solver-study question and is left for future work.
}
\section{Conclusion}
\label{sec:conclusion}

{\color{black}

In this work, we construct a QUBO Hamiltonian $H_{\mathrm{mod},k}$ representing the problem of finding nowhere-zero $\Zk$-flows on nonempty loopless multigraphs. We establish a bijection between the zero-energy states of $H_{\mathrm{mod},k}$ and flows on $G$ (\Cref{thm:mq-bij}), identify the number of zero-energy states with the flow polynomial $F(G;k)$ (\Cref{cor:flow-poly}), and show a lower bound of $\min(A,B)$ on the positive energies of $H_{\mathrm{mod},k}$ (\Cref{prp:gap}). We verify $H_{\mathrm{mod},k}$ on a test set of $59$ instances of $(G,k)$, covering yes-instances and no-instances, simple graphs, and multigraphs with parallel edges, and we separately test orientation, root, and penalty-weight choices on the robustness instances described in \Cref{sec:results-robustness}. We additionally benchmark the construction on snark families, namely the Petersen graph and the Isaacs flower snarks $J_5$ and $J_7$, confirming the expected no-instance behaviour for $k\le 4$ (\cref{tab:snark-benchmarks}).

This result should be read as a reduction, not as a statement about the performance of solvers. A zero-energy state of any solver whatsoever, regardless of the type of solver, indicates that $\phiflow(G)\le k$. If a solver returns only positive energies, this is not evidence that $\phiflow(G)>k$; establishing nonexistence requires the exact verifier described in \cref{sec:protocol} or an exhaustive ground-state enumeration of $H_{\mathrm{mod},k}$. Thus, within these stated limits, the Hamiltonian is independent of any particular solver, and the analyses performed here are entirely analytical and independent of any particular sampler.

Three follow-up directions are immediate. First, one can deploy the algorithm on an annealer to evaluate the success probability and time-to-solution of $H_{\mathrm{mod},k}$ on hard cubic instances, as compared to simulated annealing on the same Hamiltonian. Second, it would be of interest to study alternative encodings of $H_{\mathrm{mod},k}$, such as a signed-integer Hamiltonian $H_{\mathrm{int},k}$, or a specialised Hamiltonian $H_4$ for the case $k=4$. Third, the Hamiltonian $H_{\mathrm{mod},k}$ may be of interest to study in relation to the work of \citet{EsperetLagoutteMarseloo2025} on flow-reconfiguration; its samples of states with zero energy are candidates for the endpoints of a reconfiguration sequence within $\mathcal{F}(G,k)$. We leave these interesting directions for future work.
}

\bibliographystyle{plainnat}
\bibliography{references}

\end{document}